\pdfoutput=1
\documentclass[10pt,twocolumn,conference]{IEEEtran}
\usepackage{amsmath,amsfonts,amssymb}
\usepackage{graphicx}
\usepackage{cite}
\usepackage{multirow}
\usepackage{algorithm}
\usepackage{algorithmic}
\usepackage[center]{caption}
\usepackage{amsthm}

\IEEEoverridecommandlockouts

\newtheorem{theorem}{Theorem}
\newtheorem{definition}{Definition}

\newtheorem{proposition}{Proposition}
\newtheorem{example}{Example}
\theoremstyle{definition}		
\newtheorem{remark}{Remark}

\newcommand{\set}[1]{\mathcal{#1}}

\newcommand{\R}{\set{R}}

\title{Symmetry in Distributed Storage Systems}
\author{\authorblockN{Satyajit Thakor$^\dag$, Terence Chan$^\ddag$ and Kenneth W. Shum$^\dag$\\}
\authorblockA{$^\dag$Institute of Network Coding, The Chinese University of Hong Kong\\ 
$^\ddag$Institute for Telecommunications Research, University of South Australia\\
thakor@inc.cuhk.edu.hk, terence.chan@unisa.edu.au, wkshum@inc.cuhk.edu.hk}}

\begin{document}
\maketitle
\bibliographystyle{ieeetr}

\begin{abstract}
 The max-flow outer bound is achievable by regenerating codes for functional repair distributed storage system. However, the capacity of exact repair distributed storage system is an open problem. In this paper, the linear programming bound for exact repair distributed storage systems is formulated. A notion of symmetrical sets for a set of random variables is given and equalities of joint entropies for certain subsets of random variables in a symmetrical set is established. Concatenation coding scheme for exact repair distributed storage systems is proposed and it is shown that concatenation coding scheme is sufficient to achieve any admissible rate for any exact repair distributed storage system. Equalities of certain joint entropies of random variables induced by concatenation scheme is shown. These equalities of joint entropies are new tools to simplify the linear programming bound and to obtain stronger converse results for exact repair distributed storage systems.
\end{abstract}

\section{Introduction}\label{sec:I}

Distributed storage is a scheme to store data in network nodes in a distributed fashion to provide robustness to stored data in the event of node failure. In distributed data storage, data is encoded into $n$ pieces which are stored at $n$ different nodes in the network. A node requiring this data can reconstruct the data by connecting to any $k<n$ storage nodes. In a \emph{functional repair} distributed storage system, in the event of storage node failure, a replacement node connects to any $d, k\leq d< n$ storage nodes to construct a piece such that the data reconstruction is still possible after node replacement.  In an \emph{exact repair} distributed storage system, in the event of storage node failure, a replacement node connects to any $d$ storage nodes to reconstruct the original piece stored at the failed node.

In \cite{WuDimRam07}, (see also\cite{DimDodWuRam10}) a cut-set based outer bound, henceforth referred as max-flow bound, for distributed storage system is given. A class of codes, called regenerating codes, achieving all rates characterized by the max-flow bound for functional repair distributed storage systems are also introduced in \cite{DimDodWuRam10}; thus establishing the capacity region for functional repair distributed storage systems. For exact repair distributed storage systems, Rashmi {\em et al.} \cite{RasShaKum11} showed that all parameters for the minimum-bandwidth exact-repair regenerating code can be constructed explicitly. For minimum-storage regeneration, it is shown in \cite{CadJafMal} that there is a sequence of exact-repair regenerating codes attaining the optimal repair bandwidth asymptotically. However, it was shown 
that most of the rate points on the boundary of the max-flow bound for exact repair distributed storage systems are not achievable \cite{ShaRasKumRam12}. In a recent development \cite{Tia13}, the author completely characterizes the capacity region for exact repair distributed storage system with $n=4,k=3,d=3$.

In this paper we explore implications of symmetric structure of communication systems. In particular, we show that if a given communication system exhibits certain symmetry in structure, e.g. symmetry in encoding and decoding constraints, then entropies of certain subsets of random variables in the communication system are the same. This result is a significant new tool for proving converse theorems for communication systems with symmetric structure.

In Section \ref{sec:ERDS}, we describe model, capacity and an outer bound for exact repair distributed storage systems. In Section \ref{sec:MR} we formulate the LP bound for exact repair distributed storage systems. We then give a notion of symmetry for generic communication systems and prove equalities of entropies for subsets of random variables, satisfying symmetry, induced by a communication system. We devise a new coding scheme, called concatenation scheme, for exact repair distributed storage systems. It is shown that the concatenation schemes are sufficient for optimal rate. We also prove equality of joint entropies for subsets of random variables for concatenation scheme under permutation. The result is proved for any distributed storage system with any $n,k,d$. In independent and concurrent work \cite{Tia13}, the author observed and employed symmetry in $n=4,k=3,d=3$ exact repair distributed storage system to show a converse result. In Section \ref{sec:A}, we demonstrate by example a utility of the main result for exact repair distributed storage systems.

\section{Exact Repair Distributed Storage}\label{sec:ERDS}
\subsection{Model}
\begin{figure}[htbp]
\centering
  \includegraphics[scale=0.38]{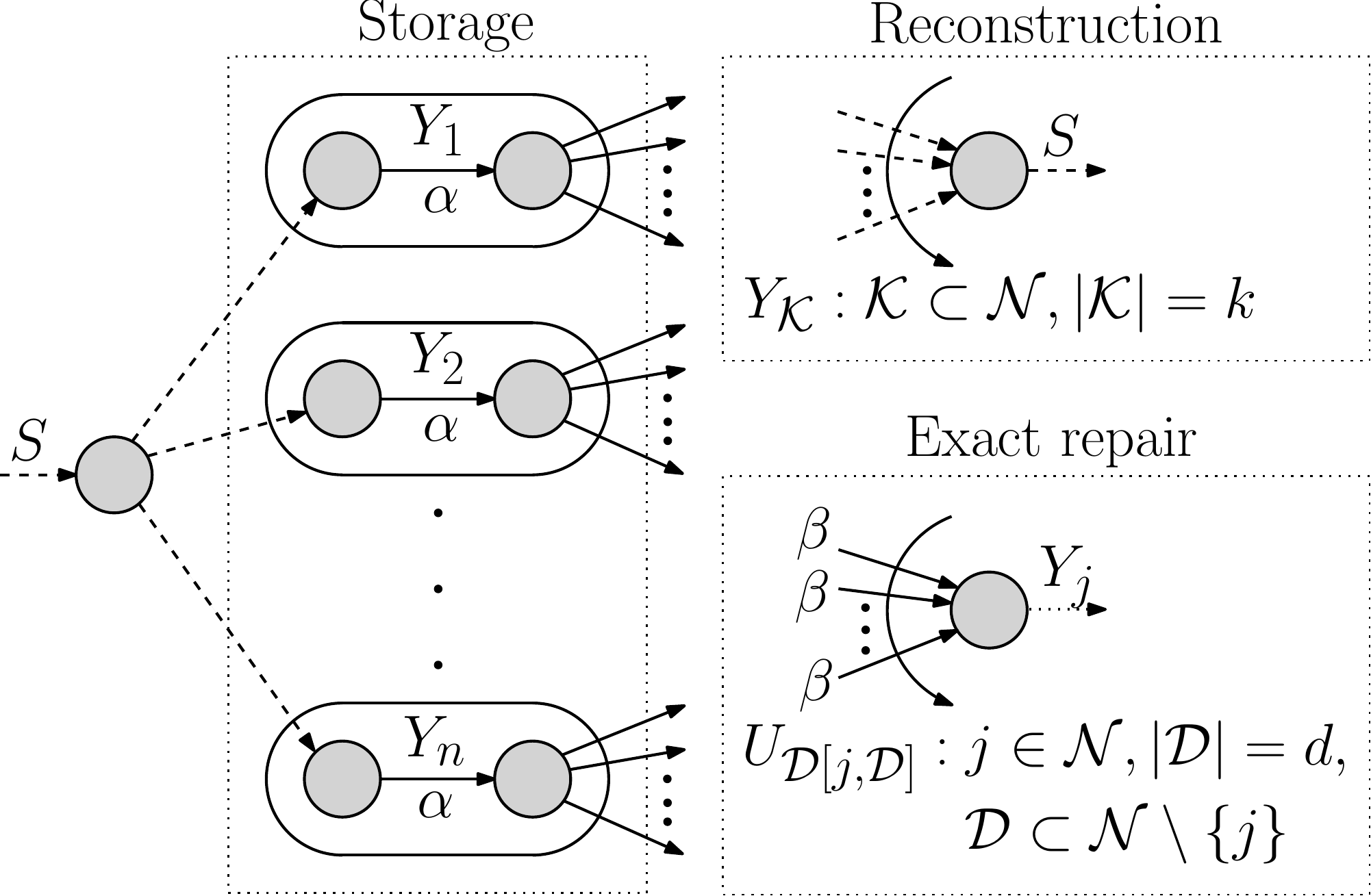}
  \caption{The exact repair distributed storage model}\label{fig:DS}
\end{figure}

Let $S$ be the uniformly distributed random variable associated with the source data to be stored distributively. The support of $S$ is denoted by $\mathcal S$. The rate of source data is $R=\log |\mathcal S|=H(S)$. Let $\mathcal N=\{1,\ldots,n\}$ be the set of storage nodes. Each node has identical storage capacity $\alpha>0$. This fact is depicted in Figure \ref{fig:DS} by two nodes within a storage node and an edge with capacity $\alpha$ connecting two nodes. The random variable $S$ is directly available to all storage nodes and this is depicted in the figure by infinite capacity directed edges (dotted) from the node where the source $S$ is available to the storage nodes. The random variable stored at node $i\in\mathcal N$ is denoted by $Y_i$ and its support is denoted by $\mathcal Y_i$.

A requirement of the distributed storage system is that $S$ must be reconstructible by any node $i \not \in \mathcal N$ in the network by accessing any set of random variables $Y_{\mathcal K}\triangleq\{Y_j, j \in \mathcal K\}, \mathcal K \subset \mathcal N$ from $|\mathcal K|=k$  storage nodes. In addition, the exact repair storage system requires that, in the case of failure of any storage node $j \in \mathcal N$, a replacement storage node must be able to reconstruct $Y_j$ by accessing a set of random variables $U_{\mathcal D[j,\mathcal D]}\triangleq \{U_{i[j,\mathcal D]},i\in \mathcal D\}$ via the set of nodes $\mathcal D \subset \mathcal N \setminus \{j\}, |\mathcal D|=d$ using directed edges of capacity $\beta>0$. The set of nodes $\mathcal D$ participating in the repair of any node is referred as the set of \emph{helper nodes}. Each helper node participating in the repair process transmits a random variable along an edge of capacity $\beta$. Let $U_{i[j,\mathcal D]}$ denote a random variable transmitted by the storage node $i$ for repair of the storage node $j$ when the set of helper nodes is $\mathcal D$. the support of $U_{i[j,\mathcal D]}$ is denoted by $\mathcal U_{i[j,\mathcal D]}$. Note that, $i \in \mathcal D, \mathcal D \subset \mathcal N \setminus \{j\}$. For distributed storage systems with $d=n-1$, there is only one set $\mathcal D$ for given $i,j$ in $U_{i[j,\mathcal D]}$ and hence the notation $U_{i[j,\mathcal D]}$ is simplified to $U_{i[j]}$. 
For simplicity, any reconstruction node is labeled by $n+1$ and any replacement node for the failed storage node $i$ is labeled by $i$. A distributed storage system with parameters $n,k,d,\alpha$ and $\beta$ is referred as $(n,k,d,\alpha,\beta)$\emph{-distributed storage system}.

\subsection{Capacity}
\begin{definition} [Exact repair code]
An exact repair code $(\Phi,\Psi)$ for a given distributed storage system is described by the sets of its encoding functions $\Phi \triangleq \{\phi_{Y_i},\phi_{U_{i[j,\mathcal D]}}\}$ and decoding functions $\Psi \triangleq \{\psi_{Y_j},\psi_{S}\}$:
\begin{align}\label{eq:DScode}
\phi_{Y_i}&: \mathcal S \rightarrow \mathcal Y_{i} , i \in \mathcal N\\
\phi_{U_{i[j,\mathcal D]}}&: \mathcal Y_i \rightarrow \mathcal U_{i[j,\mathcal D]} , i \in \mathcal D, \mathcal D \subset \mathcal N  \setminus \{j\}, |\mathcal D|=d\\
\psi_{Y_j}&: \prod_{i \in \mathcal D} \mathcal U_{i[j,\mathcal D]} \rightarrow \mathcal Y_{j} , j \in \mathcal N, \mathcal D \in \mathcal N \setminus \{j\}, |\mathcal D|=d\\
\psi_{S}&: \prod_{i \in \mathcal K} \mathcal Y_{i} \rightarrow \mathcal S, \mathcal K \subset \mathcal N, |\mathcal K|=k
\end{align}
where $\phi_{Y_i}$ 
are storage encoding functions, $\phi_{U_{i[j,\mathcal D]}}$ 
are repair encoding functions, $\psi_{Y_j}$ 
are exact repair decoding function and $\psi_{S}$ 
are reconstruction decoding functions.

\end{definition}

In practice, it is desired that the probability of error is zero. Below we define admissible rate for exact repair codes for distributed storage systems, i.e., the feasible rate for zero-error exact repair codes.

\begin{definition}[Admissible rate]
Consider a given distributed storage system with random variable $S$ associated with the data to be stored distributively at $n$ storage nodes with storage capacity $\alpha$. Let $\beta$ be the repair link capacity for the distributed storage system.  An information rate $R$ is called admissible if there exists a sequence of exact repair codes $(\Phi^{(m)},\Psi^{(m)}), m=1,2,...$ such that 
\begin{align}
\lim_{m \rightarrow \infty} \frac{\log |\mathcal Y^{(m)}_{i}|}{\log |\mathcal S^{(m)}|} \leq & \hspace{1mm}\alpha, \forall i \in \mathcal N\\
\lim_{m \rightarrow \infty} \frac{\log |\mathcal U^{(m)}_{i[j,\mathcal D]}|}{\log |\mathcal S^{(m)}|} \leq &\hspace{1mm} \beta, \forall j \in \mathcal N, i \in \mathcal D, 
\mathcal D \subset \mathcal N \setminus \{j\}\\
\mathrm{Pr}\{\psi_{S}(Y_{\mathcal K}^{(m)}) \neq S^{(m)}\} = &\hspace{1mm} 0, \forall \mathcal K \subset \mathcal N, |\mathcal K|=k\\
\mathrm{Pr}\{\psi_{Y_j}(U_{\mathcal D[j,\mathcal D]}^{(m)}) \neq Y_{j}^{(m)}\} = &\hspace{1mm} 0, \forall j \in \mathcal N,\mathcal D \subset \mathcal N\setminus\{j\},\nonumber\\
& \hspace{6mm}
|\mathcal D|=d
\end{align}
where $\psi_{S}(Y_{\mathcal K}^{(m)})$ is the decoded estimate of $S^{(m)}$ from $Y_{\mathcal K}^{(m)}$ via the mapping $\psi_{S}$ and $\psi_{Y_j}(U_{\mathcal D[j,\mathcal D]})$ is the decoded estimate of $Y_j$ from $U_{\mathcal D[j,\mathcal D]}$ via the mapping $\psi_{Y_j}$.
\end{definition}

The zero-error capacity region is the set of all admissible rates.

\subsection{Max-flow Bound}
The max-flow bound \cite{Yeu08} establishes that the achievable information rate of any network cannot be greater than the minimum of total capacity of set of edges (referred as min-cut) separating the source and the sink node. Below is the max-flow bound for distributed storage systems.
\begin{theorem}[Max-flow bound, \cite{WuDimRam07},\cite{DimDodWuRam10}]\label{thm:cut-setbound}
For a given $(n,k,d,\alpha,\beta)$-distributed storage system 
\begin{align}
R \leq \sum_{i=0}^{k-1}\min\{\alpha,(d-i)\beta\}.
\end{align}
\end{theorem}

It has been shown that the max-flow bound is not admissible for exact repair \cite{ShaRasKumRam12} and hence tighter outer bounds may be obtained using information inequalities. The basic inequalities for entropy are polymatroid axioms and the minimal set \cite{Yeu97} of basic inequalities is called the elemental set. The elemental inequalities for polymatroidal function $H$ over a set $\mathcal V$ are
\begin{align}
H(A|\mathcal V \setminus A) &\geq 0, A \in \mathcal V \label{eq:minimal_h}\\
I(A;B|\mathcal C) &\geq 0, A \neq B, \mathcal C \subseteq \mathcal V \setminus \{A,B\}. \label{eq:minimal_I}
\end{align}
The set of all polymatroidal functions is denoted by $\Gamma$.

\section{Main Results}\label{sec:MR}

\subsection{LP Bound for Distributed Storage System}\label{sec:bounds}
An explicit (computable) outer bound called linear programming (LP) bound for multi-source multi-sink networks is characterized in \cite{Yeu08}. In this section we formulate the linear programming outer bound for the capacity of distributed data storage systems. We begin by observing the constraints induced by a $(n,k,d,\alpha,\beta)$-distributed storage system.

Note that, for any $(n,k,d,\alpha,\beta)$-distributed storage system, the random variable $Y_i$ stored at node $i, i \in \mathcal N$ is a function of the data random variable $S$ and these constraints are referred as storage encoding constraints \eqref{eq:StrgEnCnstr}. The rate of any random variable $Y_i$ is less than the capacity of storage nodes; these are called storage capacity constraints \eqref{eq:StrgCCnstr}. Similarly, a random variable $U_{i[j,\mathcal D]}$ generated at the node $i$ is a function of  the random variable $Y_i$ stored at $i$. These are the repair encoding constraints \eqref{eq:RprEnCnstr}. The rate of a random variable $U_{i[j,\mathcal D]}$ cannot exceed the capacity of repair links and this fact is referred as repair link capacity constraints \eqref{eq:RprCCnstr}. The exact repair decoding constraints are \eqref{eq:RprDCnstr} describing the requirement that the random variables $Y_j$ are reconstructible from the set of random variables $U_{\mathcal D[j,\mathcal D]}$. Finally, the reconstruction constraints are \eqref{eq:RecnDCnstr} describing the requirement that the data random variable $S$ is reconstructible from any $k$ storage random variables $Y_{\mathcal K}$.
\begin{align}
H(Y_i|S)&=0, \forall i\in \mathcal N \label{eq:StrgEnCnstr}\\
H(Y_i) &\leq \alpha, \forall i \in \mathcal N \label{eq:StrgCCnstr}\\
H(U_{i[j,\mathcal D]}|Y_i)&=0, \forall i\in \mathcal D, \mathcal D \subset \mathcal N \setminus \{j\} \label{eq:RprEnCnstr}\\
H(U_{i[j,\mathcal D]}) &\leq \beta, \forall i\in \mathcal D, \mathcal D \subset \mathcal N \setminus \{j\}, j \in \mathcal N\label{eq:RprCCnstr}\\
H(Y_j|U_{\mathcal D[j,\mathcal D]})&=0, \forall j\in \mathcal N, \mathcal D \subseteq \mathcal N \setminus \{j\} \label{eq:RprDCnstr}\\
H(S|Y_{\mathcal K})&=0, \forall \mathcal K \subseteq \mathcal N \label{eq:RecnDCnstr}
\end{align}

\begin{definition}\label{def:LPBound}
Given a $(n,k,d,\alpha,\beta)$-distributed storage system with random variables
\begin{gather*}
S,\\
Y_1,\ldots,Y_n,\\
U_{i[j,\mathcal D]}, j \in \mathcal N, i \in \mathcal D \subseteq \mathcal N \setminus\{j\}.
\end{gather*}
and the set of polymatroidal entropy functions $\Gamma$ on the random variables, let $\R_{{\text{LP}}}$ be the set of source polymatroidal entropy functions $H(S)$ (non-negative real numbers) for which there exists $H\in\Gamma$ satisfying \eqref{eq:StrgEnCnstr}-\eqref{eq:RecnDCnstr}.
\end{definition}

\begin{theorem}\label{thm:LPOB}
The region $\R_{{\text{LP}}}$ is an outer bound for the set of admissible rates of exact repair distributed storage systems.
\end{theorem}

The approach used in proving the LP bound for multi-source multi-sink networks \cite{Yeu08} can be directly extended to prove Theorem \ref{thm:LPOB} for exact repair distributed storage systems. The LP bound can be computed by solving a linear program. The region $\R_{\emph{\text{LP}}}$ is projection of the region characterized by the elemental inequalities and the constraints \eqref{eq:StrgEnCnstr}-\eqref{eq:RecnDCnstr} on to the coordinate associated with the source entropy $H(S)$. In Definition \ref{def:LPBound}, by fixing the parameters $\alpha$ and $\beta$, we obtain the LP outer bound on admissible rate. But in general one can fix any two of the parameters $R, \alpha, \beta$ and obtain the LP outer bound on the admissible value of the third parameter.

In its standard form, the LP bound has an exponential number of dimensions and inequality constraints for a given set of random variables. One approach to reduce the size of the problem is to exploit functional dependence relationships among the set of random variables \cite{ThaGraCha09}, \cite{ThaGraCha11}.
\subsection{Symmetrical Sets}
In this section, we give a notion of symmetrical sets for generic communication system with $n$ random variables and show equalities of joint entropies for certain subsets of random variables. Firstly we define symmetry for a set of subsets of a given set. The notion is then extended for a set of subsets of random variables. A new set of auxiliary random variables is constructed using permutation functions. Equality of joint entropies for the auxiliary random variables are proved under the permutation mapping.

Let $\mathcal N =\{1,\ldots,n\}$. Let $\sigma$ be a permutation mapping
\begin{eqnarray}
\sigma : \{1,\ldots,n\}\rightarrow \{1,\ldots,n\}
\end{eqnarray}
and $\Sigma$ be the set of all possible permutations and hence $|\Sigma|=n!$. 
For any subset $\alpha$ of $\mathcal N$, we can define
\begin{eqnarray}
\sigma(\alpha)\triangleq \{\sigma(i): i \in \alpha\}.
\end{eqnarray}
We further define the following sets.
\begin{align}
\mathcal J_{r} &\triangleq \{\alpha: \alpha \subseteq \mathcal N, |\alpha|=r\}\\
\mathcal J &\triangleq \bigcup _{r=1}^{n}\mathcal J_{r}
\end{align}
\begin{definition}[Variable symmetry]
Let $\mathcal M \subseteq \mathcal J$. Then $\mathcal M$ is called symmetric if for all permutations $\sigma\in \Sigma$ and $\alpha \in \mathcal M, \sigma(\alpha) \in \mathcal M$.
\end{definition}
For any $\mathcal A \subseteq \mathcal M$ and $\sigma \in \Sigma$, $\sigma(\mathcal A)$ is a subset of $\mathcal M$ defined as follows.
\begin{align}
\sigma(\mathcal A) \triangleq \{\sigma (\alpha): \alpha \in \mathcal A\}.
\end{align}
From now on, we will assume that $\mathcal M$ is symmetric. In addition, we will consider a set of random variables
\begin{align}
X_{1}, \ldots, X_{n}
\end{align}
and joint random variables
\begin{align}
X_{\alpha}=(X_{i}:i \in \alpha), \alpha \in \mathcal M.
\end{align}

Let $X_{\mathcal M}$ denote the symmetric set of random variables $\{X_{\alpha}, \alpha \in \mathcal M\}$.
For each $\alpha \in \mathcal M$ consider independent and identically distributed (i.i.d.) random variables
\begin{align}
X_{\sigma(\alpha)}^{\sigma}, \sigma \in \Sigma
\end{align}
such that the joint probability distribution of $X_{\mathcal A}=\{X_{\alpha},\alpha \in \mathcal A\}$ is the same as $X_{\sigma(\mathcal A)}^{\sigma}=\{X_{\sigma(\alpha)}^{\sigma},\alpha \in \mathcal A\}$ and independent for all possible permutations.
Hence
\begin{align}
H(X_{\alpha})&=H(X_{\sigma(\alpha)}^{\sigma}), \sigma \in \Sigma\\
H(X_{\alpha}, \alpha \in {\mathcal A})&=H(X_{\sigma(\alpha)}^{\sigma}, \alpha \in \mathcal A),\sigma \in \Sigma\\
H(X_{\sigma(\alpha)}^{\sigma}, \alpha \in \mathcal A, \sigma \in \Sigma)&=\sum_{\sigma \in \Sigma} H(X_{\sigma(\alpha)}^{\sigma}, \alpha \in \mathcal A).
\end{align}

\begin{remark}
In the notation $X_{\sigma(\alpha)}^{\sigma}$, the superscript signifies that for each permutation the random variable is distinct even though the subscript may be the same. For example, let $\sigma, \tau \in \Sigma$ 
then the random variables $X_{\sigma(\alpha)}^{\sigma}$ and $X_{\tau(\alpha)}^{\tau}$ are distinct random variables. That is, $X_{\sigma(\alpha)}^{\sigma}$ and $X_{\tau(\alpha)}^{\tau}$ have identical distribution (same as that of $X_{\alpha}$) but the distinct superscripts suggest that they are independent.
\end{remark}

Now, construct auxiliary random variables as follows.
\begin{align}
W_{\alpha}= (X_{\sigma(\alpha)}^{\sigma}, \sigma \in \Sigma), \forall \alpha \in \mathcal M
\end{align}

\begin{proposition}\label{thm:EntropyEq}
Suppose the set of random variables $\{W_{\alpha},\alpha \in \mathcal M\}$ is defined as above. Then for any $\mathcal A \subseteq \mathcal M$ and $\sigma \in \Sigma$
\begin{align}
H(W_{\alpha}, \alpha \in \mathcal A)= H(W_{\alpha}, \alpha \in \sigma(\mathcal A)).
\end{align}
\end{proposition}

\begin{proof}
\begin{align}
H(W_{\alpha}, \alpha \in \mathcal A)&=H(X_{\zeta(\alpha)}^\zeta, \zeta \in \Sigma)\\
&=\sum_{\zeta \in \Sigma} H(X_{\zeta(\alpha)}^\zeta, \alpha \in \mathcal A)\\
&=\sum_{\mu \in \Sigma} H(X_{\mu\sigma(\alpha)}^{\mu\sigma}, \alpha \in \mathcal A)\label{eq:eq1thm1}\\
&=\sum_{\mu \in \Sigma} H(X_{\mu(\beta)}, \beta \in \sigma(\mathcal A))\\
&=H(W_{\alpha}, \alpha \in \sigma(\mathcal A))
\end{align}
where \eqref{eq:eq1thm1} follows by letting $\mu=\zeta \sigma^{-1}$.
\end{proof}
\subsection{Concatenation Scheme for Distributed Storage}

Now we propose a new coding scheme, called concatenation scheme, for distributed storage systems. We show that the concatenation scheme is sufficient to achieve any admissible rate. We also show that the concatenation schemes have symmetrical properties and hence entropies for certain subsets of random variables induced by concatenation scheme are equal.

Let there be a $(n,k,d,\alpha,\beta)$-distributed storage system with random variables
\begin{gather*}
S,\\
Y_1,\ldots,Y_n,\\
U_{i[j,\mathcal D]}, j \in \mathcal N, i \in \mathcal D \subseteq \mathcal N \setminus\{j\}.
\end{gather*}
Note that, the random variables are induced by a scheme or a code. Then, there exist another scheme, called a \emph{permutation scheme}, for the $(n,k,d,\alpha,\beta)$-distributed storage system with random variables as follows.
\begin{gather*}
S^{\sigma},\\
Y_{\sigma(1)}^{\sigma},\ldots,Y_{\sigma(n)}^{\sigma},\\
U_{\sigma(i)[\sigma(j),\sigma(\mathcal D)]}, \sigma(j) \in \mathcal N, \sigma(i) \in \sigma(\mathcal D) \subseteq \mathcal N \setminus\{\sigma(j)\}.
\end{gather*}
where $\sigma$ is a permutation function defined by mapping $\sigma: \mathcal N\rightarrow \mathcal N$.  Let $\Sigma$ be the set of all possible permutation functions for the set $\mathcal N$.

\begin{definition}
Given a distributed storage scheme with random variables $S,Y_1,\ldots,Y_n, U_{i[j,\mathcal D]}, j \in \mathcal N, i \in \mathcal D \subseteq \mathcal N \setminus\{j\}$,  the concatenation scheme is designed by combining all possible permutation schemes as follows.
\begin{align}
V&=(S^{\sigma},\sigma\in \Sigma)\\
A_{i}&=(Y_{\sigma(i)}^{\sigma}, \sigma \in \Sigma)\\
B_{i[j,\mathcal D]}&= (U_{\sigma(i)[\sigma(j),\sigma(\mathcal D)]}^{\sigma}, \sigma \in \Sigma)
\end{align}
where $S^\sigma, \sigma \in \Sigma$ are i.i.d. random variables with the distribution of random variable $S$.
\end{definition}

\begin{example}[Concatenation scheme]
We describe a concatenation scheme for $(3,2,2,\alpha,\beta)$-distributed storage system with the following random variables.
\begin{gather*}
S,
Y_1,Y_2,Y_3,\\
U_{1[2]},U_{1[3]},U_{2[1]},U_{2[3]},U_{3[1]},U_{3[2]}
\end{gather*}

Note that, given a coding scheme for $(3,2,2,\alpha,\beta)$-distributed storage system, six distinct permutation schemes are possible. Let $\Sigma=\{\textsf{1},\ldots,\textsf{6}\}$ be the set of all permutation functions. The concatenation scheme has following random variables.
\begin{align}
V= & (S^{\textsf{1}},\ldots,S^{\textsf{6}})\\
A_i= & (Y_{\textsf{1}(i)}^{\textsf{1}},\ldots,Y_{\textsf{6}(i)}^{\textsf{6}}), i\in\{1,2,3\}\\
B_{i[j]}= & (U_{\textsf{1}(i)[\textsf{1}(j)]}^{\textsf{1}},\ldots,U_{\textsf{6}(i)[\textsf{6}(j)]}^{\textsf{6}}),\nonumber\\
  & i\in\{1,2,3\}, j \in \{1,2,3\}\setminus \{i\}
\end{align}
\end{example}

Now we show that, for any admissible rate $R$ in the zero-error capacity region, it is sufficient to consider concatenation scheme.

\begin{theorem}\label{thm:sufficiency}
If the rate $R$ is admissible for a given $(n,k,d,\alpha,\beta)$-distributed storage system, then there exists a concatenation scheme with admissible rate $R$.
\end{theorem}
\begin{proof}
Let $R$ be the rate admissible by a sequence of exact repair codes $(\Phi^{(m)},\Psi^{(m)})$ for a given $(n,k,d,\alpha,\beta)$-distributed storage system with random variables $S, Y_1,\ldots,Y_n, U_{i[j,\mathcal D]}, j \in \mathcal N, i \in \mathcal D \subseteq \mathcal N \setminus\{j\}$. Then, the concatenation scheme can be designed with random variables
\begin{align}
V&=(S^{\sigma},\sigma\in \Sigma)\\
A_{i}&=(Y_{\sigma(i)}^{\sigma}, \sigma \in \Sigma)\\
B_{i[j,\mathcal D]}&= (U_{\sigma(i)[\sigma(j),\sigma(\mathcal D)]}^{\sigma}, \sigma \in \Sigma)
\end{align}
Note that the concatenation scheme with random variables $V, A_{i},B_{i[j,\mathcal D]}$ is a code for a $(n,k,d,n! \alpha,n! \beta)$-distributed storage system with source information rate $n! \times R$ designed by concatenating all possible permutation schemes.
\end{proof}

In Theorem \ref{thm:sufficiency}, the structure of a concatenation scheme allows us to use an argument in essence similar to time-sharing for channel codes or space-sharing for storage codes.

\begin{theorem}
For a concatenation scheme the following equalities are true for all $\sigma \in \Sigma$.
\begin{align}
H(A_{\gamma})&=H(A_{\sigma(\gamma)}) \label{eq:symmetry1}
\end{align}
\begin{align}
H(B_{\delta})&=H(B_{\sigma(\delta)}) \label{eq:symmetry2}\\
H(A_{\gamma},B_{\delta})&=H(A_{\sigma(\gamma)},B_{\sigma(\delta)})\label{eq:symmetry3}
\end{align}
where $ \gamma \subset \mathcal N$, $\delta \subset \Delta$ and
\begin{align}
\Delta \triangleq \{i[j,\mathcal D] \in \mathcal N \times \mathcal N \times 2^\mathcal N: j \in \mathcal N, i \in \mathcal D \subset \mathcal N \setminus \{j\}\}.\nonumber
\end{align}
\end{theorem}

\begin{proof}
The power set of the set of random variables in a concatenation scheme is a symmetrical set. The equalities then follows by Proposition \ref{thm:EntropyEq} for symmetrical sets.
\end{proof}
\section{Applications}\label{sec:A}
\subsection{Simplification of the LP Bound}
In this section, we demonstrate an example of reduction in the number of dimensions of LP bound computation for distributed storage systems using symmetry result for a concatenation scheme.
\begin{example}[The LP bound dimensions reduction for $(3,2,2,\alpha,\beta)$-distributed storage system]
The first reduction is obtained using functional dependence relationships among random variables \cite{ThaGraCha09}, \cite{ThaGraCha11}. Below are the maximal irreducible sets \cite{ThaGraCha09} for the distributed storage system.
\begin{gather*}
\{S\},\{Y_1,Y_2\},\{Y_1,Y_3\},\{Y_2,Y_3\},\{Y_1,U_{2[3]}\},\{Y_1,U_{3[2]}\},\\
\{Y_2,U_{1[3]}\},\{Y_2,U_{3[1]}\},\{Y_3,U_{1[2]}\},\{Y_3,U_{2[1]}\},\\
\{U_{2[1]},U_{3[1]},U_{2[3]}\},\{U_{2[1]},U_{3[1]},U_{3[2]}\},\\
\{U_{1[2]},U_{3[2]},U_{1[3]}\},\{U_{1[2]},U_{3[2]},U_{3[1]}\},\\
\{U_{1[3]},U_{2[3]},U_{1[2]}\},\{U_{1[3]},U_{2[3]},U_{2[1]}\}
\end{gather*}
Other sets which are irreducible but not maximal and not subset of any irreducible set are as follows.
\begin{gather*}
\{Y_1,U_{2[1]}\},\{Y_1,U_{3[1]}\},\{Y_2,U_{1[2]}\},\{Y_2,U_{3[2]}\},\\
\{Y_3,U_{1[3]}\},\{Y_3,U_{2[3]}\},\{U_{1[2]},U_{2[3]},U_{3[1]}\},\\
\{U_{1[3]},U_{2[1]},U_{3[2]}\}
\end{gather*}
All maximal irreducible sets have the same entropy and for every reducible set, there exists a strict subset which is irreducible and have the same entropy as the reducible set \cite{ThaGraCha09}. Hence, the dimensions of the LP problem are any one maximal irreducible set and all other non-maximal irreducible sets. There are 30 such sets:
\begin{gather*}
\{S\},\{Y_1\},\{Y_2\},\{Y_3\},\{U_{2[3]}\},\{U_{3[2]}\},\{U_{1[3]}\},\{U_{3[1]}\},\\
\{U_{1[2]}\},\{U_{2[1]}\},\{U_{2[1]},U_{3[1]}\},\{U_{2[1]},U_{2[3]}\},\{U_{3[1]},U_{2[3]}\},\\
\{U_{2[1]},U_{3[2]}\},\{U_{3[1]},U_{3[2]}\},
\{U_{1[2]},U_{3[2]}\},
\{U_{1[2]},U_{1[3]}\},\\
\{U_{3[2]},U_{1[3]}\},
\{U_{1[2]},U_{3[1]}\},
\{U_{1[3]},U_{2[3]}\},
\{U_{2[3]},U_{1[2]}\},\\
\{U_{1[3]},U_{2[1]}\},
%
\{Y_1,U_{2[1]}\},
\{Y_1,U_{3[1]}\},\{Y_2,U_{1[2]}\},\\
\{Y_2,U_{3[2]}\},\{Y_3,U_{1[3]}\},\{Y_3,U_{2[3]}\},
\{U_{1[2]},U_{2[3]},U_{3[1]}\},\\\{U_{1[3]},U_{2[1]},U_{3[2]}\}
\end{gather*}
Now, using the sufficiency and symmetry results for concatenation scheme, the number of dimensions for the LP problem reduces to 9:\vspace{-2mm}
\begin{gather*}
\{S\},\{Y_1\},\{U_{1[2]}\},\{U_{1[2]},U_{3[2]}\},\{U_{1[2]},U_{1[3]}\},\\
\{U_{1[2]},U_{3[1]}\},
\{U_{1[2]},U_{2[3]}\},\{Y_1,U_{2[1]}\},\{U_{1[2]},U_{2[3]},U_{3[1]}\}
\end{gather*}
\end{example}
\subsection{Proving Converse Theorems}
As we already mentioned in Section \ref{sec:I}, symmetry in $(4,3,3,\alpha,\beta)$-exact repair distributed storage \cite{Tia13} is observed and utilised to show a converse result. We strongly believe that symmetry results proved in this paper for $(n,k,d,\alpha,\beta)$-exact repair distributed storage systems will be useful to prove new converse results for general $n,k,d$. One of our future direction of research is to investigate use of symmetry results to analytically derive stronger converse theorems for general $(n,k,d,\alpha,\beta)$-exact repair distributed storage systems.
\section{Conclusion}
We proposed a concatenation scheme with symmetric structure for exact repair distributed storage and shown its sufficiency for any admissible rate. Using symmetry in the concatenation scheme we proved equalities of joint entropies under permutation for subset of random variables in concatenation scheme. One application of the main results is to reduce the size of the LP bound and is demonstrated by an example. The equalities proved for exact repair distributed storage are important new tools and future direction of research is to employ these tools to derive better outer bounds on the capacity of general $(n,k,d,\alpha,\beta)$-exact repair distributed storage systems.
\section*{Acknowledgement}
The authors thank Prof. Raymond W. Yeung for the useful discussions. The work of Satyajit Thakor and Kenneth W. Shum was partially supported by a grant from the University Grants Committee of the Hong Kong Special Administrative Region, China (Project No. AoE/E-02/08). The work of Terence Chan was  supported by the Australian Government under ARC grant DP1094571.
\bibliographystyle{IEEEtran}
\bibliography{DSbib}
\end{document}